\documentclass{llncs}

\makeatletter\let\proof\@undefined\let\endproof\@undefined\makeatother
\usepackage{amsthm}

\usepackage{amssymb}
\usepackage{latexsym}

\usepackage[mathscr]{eucal}
\usepackage{graphicx}
\usepackage[colorlinks=true,citecolor=black,linkcolor=black,urlcolor=red]{hyperref}
\usepackage{caption}

\newcommand{\Z}{\mathbb{Z}}
\newcommand{\F}{\mathbb{F}}

\usepackage[mathscr]{eucal}
\usepackage{graphicx}
\usepackage[colorlinks=true,citecolor=black,linkcolor=black,urlcolor=red]{hyperref}

\def \F {{\mathbb F}}

\def \Z {{\mathbb Z}}
\newcommand{\rmv}[1]{}
\newcommand{\change}[1]{{{\color{black}#1}}}
\newcommand{\modify}[1]{{{\color{black}#1}}}
\newcommand{\revise}[1]{{{\color{black}#1}}}

\newcommand{\ChangeRed}[1]{{{\color{black}#1}}}
\newcommand{\Change}[1]{{{\color{black}#1}}}
\newcommand{\ChangeGreen}[1]{{{\color{black}#1}}}

\newcommand{\ChangeMS}[1]{{{\color{black}#1}}}
\newcommand{\ChangeW}[1]{{{\color{black}#1}}}

\begin{document}
\thispagestyle{empty}
\title{On the stability of periodic binary sequences with zone restriction}

\author{Ming Su\inst{1} \space and \space  Qiang Wang\inst{2}}
\institute{Department of Computer Science, Nankai University, Tianjin, China\\
\email{nksuker@gmail.com}
\thanks{Ming Su is supported by the China Scholarship Council.}
 \and School of Mathematics and
Statistics, Carleton University, Ottawa, Canada\\
\email{wang@math.carleton.ca}
\thanks{The research of Qiang Wang is partially funded by NSERC of Canada.}
}

\maketitle

\begin{abstract}
Traditional global stability measure for sequences is hard to
determine because of large search space. We propose the $k$-error
linear complexity with a zone restriction for measuring the local
stability of sequences.
\ChangeW{
Accordingly, we can efficiently
determine the global stability by studying a local stability for these sequences.
 For several classes of sequences,   we demonstrate that the $k$-error linear complexity is
identical to the $k$-error linear complexity within  a zone, while the
length of a zone is much smaller than the whole period when the $k$-error
linear complexity is large.    These sequences  have periods $2^n$, or $2^v r$ ($r$ odd prime and $2$ is primitive modulo $r$), or $2^v p_1^{s_1} \cdots p_n^{s_n}$ ($p_i$ is an odd prime and $2$ is primitive modulo $p_i$ and $p_i^2$,  where $1\leq i \leq n$) respectively. In particular,  we completely determine the spectrum of $1$-error linear
complexity with any zone length for an arbitrary $2^n$-periodic
binary sequence.}

\end{abstract}

\begin{keywordname}
Stability $\cdot$ Linear complexity $\cdot$ $k$-linear complexity
$\cdot$ Zone restriction
\end{keywordname}

\section{Introduction}

Let $S=(s_0,s_1,s_2,\ldots)$ be an $N$-periodic sequence with terms
in the finite field $\mathbb{F}_q$ of $q$ elements. We note that $N$
need not be the least period of the sequence. We  denote
$S=(s_0,s_1,\ldots,s_{N-1})^{\infty}$ and define
\change{$S^N(x)=s_0+s_1 x+\cdots+s_{N-1}x^{N-1}$}. The \emph{linear
complexity} of a periodic sequence over $\mathbb{F}_q$ is the length
of the shortest linear recurrence relation which the sequence
satisfies.
In algebraic terms the linear complexity of an $N$-periodic sequence
is given by $L(S)=N-\deg(\gcd(1-x^N,\ChangeMS{S^N(x)}))$; see for
example \cite[p. 28]{TDR}.

\modify{ For an integer $k$, $0 \leq k \leq N$, the minimum linear
complexity of those sequences with not more than $k$ term changes in
a period $N$ from the original sequence $S$ is called the
\emph{$k$-error linear complexity} of $S$, denoted as $L_{N,k}(S)$,
i.e.,
 \begin{displaymath}
 L_{N,k}(S)= \min_{ W_H(T)\leq k }  \{L(S+T) \},
 \end{displaymath}
 where $T$ is an $N$-periodic sequence,
 $W_H(T)$ is the Hamming weight of $T$ in one
 period,   the
addition ``$+$" for two sequences is defined elementwise in
$\mathbb{F}_q$. A  sequence $T$ reaching the  $L_{N,k}(S)$ is called
an \emph{error vector} of the $k$-error linear complexity. When
$N=2^n$, we denote the $k$-error linear complexity of $S$ by
$L_k(S)$. }

\ChangeW{In addition to the Berlekamp-Massey algorithm
\cite{Massey} for computing the linear complexity with computational
complexity $O(N^2)$, there are efficient algorithms of several types
of periodic sequences with computational complexity $O(N)$, such as
the Games-Chan algorithm \cite{GC83} for computing the linear
complexity of a $2^n$-periodic binary sequence; the algorithm due to
 Meidl  \cite{MeidlDCC08} for computing the linear complexity of a $u
2^n$-periodic binary sequence, where $u$ is odd ;
the algorithm for computing the linear complexity of a sequence with
period $p^n$ over $\F_q$ \cite{XiaoWLI2000}, where $p$ is an odd
prime and $q$ is a prime and a primitive root $\bmod \; p^2$; and
the algorithm for computing the linear complexity of a sequence with
period $2p^n$ over $\F_q$ \cite{WeiXC02}, where $p$ and $q$ are odd
primes, and $q$ is a primitive root $\bmod \; p^2$. These algorithms work because the
factorization of $X^N-1$ is simple under these assumptions.
Comparatively, there are also   efficient algorithms of computing the $k$-error
linear complexity for certain types of sequences such as the Stamp-Martin algorithm
\cite{StampMartin} for computing the $k$-error linear complexity of
a $2^n$-periodic binary sequence, the algorithm for computing the
$k$-Error Linear Complexity of $p^n$-periodic sequences over
$\F_{p^m}$ \cite{KaidaUI}, and the algorithm for computing the
$k$-error linear complexity of a sequence with period $2p^n$ over
$\F_q$ \cite{ZhouJQ}.   We also remark that there are some  studies on
 the properties of $k$-error linear complexity of binary
 sequences, see \cite{HanCY07}, \cite{SM_IEICE}.
 Earlier, S$\breve{a}$l$\breve{a}$gean et al. studied approximation
algorithms for the $k$-error linear complexity of odd-periodic
binary sequences by using DFT and some relaxation
\cite{AS-ISIT,AS-SETA}. However, there is no efficient algorithm for calculating the
 $k$-error linear complexity for an arbitrary binary sequence, in particular, binary sequence with even period.  }



 A well-designed sequence should not only  have a
large linear complexity, but also large $k$-error linear
complexities. This means its linear complexity should not decrease a
lot when $k$ errors occur; see \cite{StampMartin} and \cite{Ding}.
 In order to measure the stability of a given periodic
sequence, we have to consider $k$ errors \Change{that can occur
anywhere within the whole} period $N$.
\Change{This means} the computational task is heavy because the
capacity of search space for all possible binary errors is
$\sum_{t=0}^k {N \choose t}$, which is very large for common $N$ and
moderate $k$. \Change{Indeed, it  becomes} exponential of $N$ when
$k$ is large, resulting in infeasible computations.
This motivates us to study $k$-error linear complexity with a zone
restriction.
\ChangeW{ Intuitively, there can be many error vectors that
reaching the $k$-error linear complexity.  We show later on that for many sequences we can find a window of proper
length $Z$ containing  at least one error
vector, no matter where  we  start with.
 For the convenience,  we may assume the zone starts at position 0 and ends at
position $Z-1$. }Therefore we define the {\it $k$-error linear
complexity with a zone of length $Z$}, denote by $(N,k; Z)$-error
linear complexity, as the minimum of all $k$-error linear complexity
such that these errors occur in positions from $0$ to $Z-1$. That
is,
\begin{eqnarray} \label{Definitoin-k;Z-LC}
 L_{N,k; Z}(S)= \min_{ W_H(T)\leq k  \atop T[i]=0, \; Z\leq i <N}  \{L(S+T) \}.
 \end{eqnarray}
Obviously, $L_{N,k; Z}(S)$ is easier to  compute and this  provides
a natural upper bound of $L_{N,k}(S)$.

\ChangeW{In this paper, we study the relation
between $L_{N,k; Z}(S)$ and $L_{N,k}(S)$ and prove that
$L_{N,k;Z}(S)=L_{N,k}(S)$ for several classes of  sequences and the
zone length $Z$ can be very small comparing to the period $N$. This
means that we can  efficiently determine  the global stability via a local
stability.  We focus on binary sequences with even period and large linear complexity,
in particular,  several classes of sequences with  periods $2^n$, or $2^v r$ ($r$ odd prime and $2$ is primitive modulo $r$), or $2^v p_1^{s_1} \cdots p_n^{s_n}$ ($p_i$ is an odd prime and $2$ is primitive modulo $p_i$ and $p_i^2$,  where $1\leq i \leq n$) respectively.

Sequences with period $2^n$ have attracted a lot of attention \cite{Paterson2009}; one typical example is the  de Bruijin sequence of maximal $2^n$-periodic sequence generated by
NFSR of stage $n$ \cite{GongG}. Despite that there is an efficient algorithm to compute $k$-error linear complexity of these sequences,  we still demonstrate our method by showing that there exists a small zone
of length $Z = 2^{\lceil \log_2 (2^n-L_k(S)) \rceil}$ containing the \emph{support} (positions of nonzero entries) of an
error vector reaching the $k$-error linear complexity for any $2^n$-periodic binary sequence. This means
that we can indeed reduce the global stability to a local stability.
Furthermore,  we completely describe the spectrum of 1-error linear
complexity with any given zone length.  This can help us to obtain the exact counting
functions for any $2^n$-periodic binary sequence.

Afterwards, we found two more classes of binary sequences such that their
global stability can be reduced to a local stability.  The first class of sequences has large linear complexity and $k$-error
linear complexity with period $2^v r$, such that $r$ is an odd prime and $2$ is a
primitive root modulo $r$. The length of a zone is $Z=2^{\lceil \log_2(N-L_{N,k}(S)) \rceil}$.
 More details can be found in Theorem~\ref{2^vr-1-zonelength}.
We want to emphasize that  our result applies to quite a lot of sequences.
  By Artin's conjecture, approximately 37\% of all primes satisfy that $2$ is a primitive
root modulo $r$. We also justify that there are high proportion of sequences who
have the required large linear complexity and $k$-error linear
complexity, among those sequences with period $2^v r$ where $2$ is primitive modulo $r$.  In particular, we show that  if the growth of $p$ is polynomial in terms of $r$ and $\sum_{t=0}^k {2^v r \choose t}< \frac{2^{r-1}}{2^v}$, then almost all these sequences have desired properties so that their global stability can be reduced to local stability.  The second class of sequences has  the period
$N=2^v p_1^{s_1} p_2 ^{s_2}\ldots p_n^{s_n}$, where $p_i$ is odd prime and
$2$ is a primitive root modulo $p_i$ and $p_i^2$ for all $1\leq i \leq
n$.
For any such $N$-periodic  binary sequence $S$ such that
$L(S)>L_{N,k}(S)\geq N-\min(2^v,p_1-2,\ldots, p_n-2)$, we  show in Theorem 4 that
 there exists a zone of length $Z=2^{\lceil \log_2(N-L_{N,k}(S))
 \rceil}$ such that $L_{N,k}(S)=L_{N,k;Z}(S)$. }

The rest of this paper is organized as follows. In Section
\ref{Sec:reduction}, we study the $(k;Z)$-error linear complexity
for any periodic binary sequence with period $2^n$, and find a
proper zone of length \modify{$Z = 2^{\lceil \log_2 (2^n-L_k(S))
\rceil} $} such that $L_{N,k}=L_{N,k;Z}$. \Change{The larger
$L_{N,k}$, the smaller zone length $Z$}. In Section
\ref{Sec:Stabistic property}, we study  the linear complexity
affected by 1-error occurrence within a zone of length $Z$, and give
the exact counting functions of the $1$-error linear complexity
with a restriction on zone length $Z$ for a random $2^n$-periodic
binary sequence.
In Section~\ref{sec:N=2^vr}, \ChangeW{we  prove Theorems 3 and Theorem 4. }



\section{Reduction from global stability to local stability with a zone
restriction for any binary sequence of period $2^n$}
\label{Sec:reduction}

In this section, we show that the global stability can be reduced to
local stability with zone restriction for any binary sequence of
period $2^n$.
\ChangeW{We denote  the binary sequence with the only nonzero entry
`$1$'  at position $j$ by $E_1(j)$, $0\leq j< 2^n$, in each period
$2^n$,  and the expected $(k;Z)$ linear complexity of $N$-periodic
sequences  by $\mathcal{E}_{N,k;Z}$}. Without causing any confusion,
we denote $E_{k;Z}$ by the expected $(k;Z)$ linear complexity of
$2^n$-periodic binary sequences, and $\mathcal{N}_k(c)$  by the
number of sequences achieving $k$-error linear complexity value $c$
of $2^n$-periodic binary sequences. Because the sequence $S$ has
period $2^n$, we only need to consider the multiplicity of $x=1$ as
a root of \Change{$S^{2^n}(x)$} when we compute the linear
complexity $L(S) = 2^n - deg(gcd(1-X^{2^n}, S^{2^n}(X))$.  It is
straightforward to derive  the following useful result.

\begin{lemma} \label{LC-AdditiveProperty}
$L(S+S^\prime)= \max\{L(S), L(S^\prime)\}$ for two $2^n$-periodic
sequences $S, S'$ if $L(S^\prime)\neq L(S)$, and
$L(S+S^\prime)<L(S)=L(S^\prime)$ if $L(S)=L(S^\prime)$. In
particular $L(E_1(j)+E_1(j+i2^s))\leq 2^n-2^s$, \Change{where $1\leq
s < n$.}
\end{lemma}
\begin{proof}
\Change{ Obviously,  we can write
$S^{2^n}(x)=(1-x)^{2^n-L(S)}g_S(x)$ for the sequence $S$, where
$g_S(1) =1$. Simlarly,
$S'^{2^n}(x)=(1-x)^{2^n-L(S)}g_{S^\prime}(x)$, where
$g_{S^\prime}(1)=1$.  If $L(S^\prime)\neq L(S)$, then
$S^{2^n}(x)+s'^{2^n}(x)=(1-x)^{2^n-\max\{L(S),L(S^\prime)\}}
\tilde{g}(x)$, and $\tilde{g}(1)=1$. Therefore we have
$L(S+S^\prime)= \max\{L(S), L(S^\prime)\}$. If $L(S^\prime)= L(S)$,
we obtain
$S^{2^n}(x)+S'^{2^n}(x)=(1-x)^{2^n-L(S)}(g_S(x)+g_{S^\prime}(x))$,
and $g_S(1)+g_{S^\prime}(1)=0$.} Therefore,
$L(S+S^\prime)<L(S)=L(S^\prime)$. In particular, the degree of
$\gcd(x^{j}+x^{j+i2^s}, x^{2^n}-1)$ is at least $2^s$ because
$x^{j}+x^{j+i2^s}=x^j(1+x^i)^{2^s}$.   \end{proof}

It is well known from  \cite{WMeidl} and \cite{SM_IEICE} that $L_k
\neq 2^n-2^s$ for any integer $s<n$, \ChangeRed{$k>0$}.
 In particular, when $k=1$, we can determine the number of error vectors $E_1(j)$ such that $L(S+E_1(j)) = L_1(S)$.

\begin{lemma} \label{ErrorVectorsSpan}
For a sequence $S$ satisfying $L(S)=2^n$ and
$2^n-2^s<L_1(S)<2^n-2^{s-1}$, we have exact $2^{n-s}$ error vectors
with Hamming weight $1$ in one period achieving $L_1(S)$.
\end{lemma}
\begin{proof}
Suppose there is $E_1(j)$ such that $0\leq j < 2^s$ and
$L(S+E_1(j))=L_1(S)$, we claim that we must have a set of error
vectors at positions $j+i2^s$, where $0\leq i <2^{n-s}$ and the
addition $+$ is performed modulo $2^n$, such that
$L(S+E_1(j+i2^s))=L(S+E_1(j)+E_1(j)+E_1(j+i2^s)) = L(S+E_1(j))$.
Since $ L(E_1(j)+E_1(j+i2^s)) \Change{\leq}  2^n-2^s < L_1(S) =
L(S+E_1(j))$, we conclude the above claim by Lemma
\ref{LC-AdditiveProperty}.

For any other error vector $E_1(j^\prime)$ such that
$j^\prime-j\equiv d ~(mod~2^s)$ and $0<d<2^s$, the largest
nonnegative integer $t$ such that $2^t \mid (j^\prime-j)$ must
satisfy $t < s$.  Hence the degree of $\gcd(x^{j}+x^{j^\prime},
x^{2^n}-1)$ is exactly $2^t \leq 2^{s-1}$ and thus $L(E_1(j) +
E_1(j^\prime)) = 2^n -2^t \geq 2^n -2^{s-1} > L_1(S)$. Therefore
$L(S+E_1(j^\prime)) =L(S+E_1(j)+E_1(j)+E_1(j^\prime)) =
L(E_1(j)+E_1(j^\prime)) > L_1(S)$. This shows that there is exactly
one error vector $E_1(j)$ such that $0\leq j < 2^s$ and $L(S+E_1(j))
= L_1(S)$. Hence there are exactly $2^{n-s}$ error vectors in the
whole period achieving $L_1(S)$.
     \end{proof}



%
%
%
%
%
%

For any positive integer $m$, we let $E_m$ denote \modify{a binary
vector of length $2^n$  with Hamming weight $m$. For example, assume
$E_m$ has `1' at positions $i_1, i_2,\ldots, i_{m}$, where $0\leq
i_1<i_2<\ldots <
 i_m \leq 2^n-1$}.  Then we define the support of $E_m$ as  $Supp(E_m)=\{ i_1, i_2,\ldots, i_m \}$.
 Now we can show that there exist at least one error vector whose support  is  contained in a smaller zone.

\begin{lemma} \label{Lem:ProperZone}
Let $S$ be a binary sequence with period $2^n$. Suppose $2^n-2^s<
L_k(S)<2^n-2^{s-1}$ for some integer $s$. Then there exists at least
one error vector $E_m$ of weight $m$, $m\leq k$ such that
$L(S+E_m)=L_k(S)$ and $supp(E_m) \subseteq [0, 2^s)$.
\end{lemma}
\begin{proof}
There exists an error vector $E_m$, $m\leq k$ of Hamming weight $m$
so that $L(S+E_m)=L_k$.
 If $i_m \geq 2^s$, then we can define a new vector  \ChangeRed{$E_{m}^{'}:=E_m+E_1(i_m)+E_1( i_m \bmod  {2^s})$}  so that
\begin{eqnarray*}
L(S+E_{m}^{'})=L \left(S+E_m + E_1(i_m)+E_1( i_m \bmod {2^s})
\right).
\end{eqnarray*}
Because
\begin{eqnarray*}
L(S+E_m)=L_k>2^n-2^s \geq L(E_1(i_m)+E_1( i_m \bmod {2^s})),
\end{eqnarray*}
we get $L(S+E_{m'})=L_k$,
where $Supp(E_m')=Supp(E_m)\setminus \{i_m\} \cup \{ i_m \bmod {2^s}
\}$.
 Therefore, we can consecutively adjust those entries of $E_m$ so that we can find
 $\bar{E}_m$ such that $supp(\bar{E}_m) \subseteq  [0, 2^s)$ and $L(S+\bar{E}_m)=L_k$.
\end{proof}

\ChangeW{
\begin{remark}
Actually,
Lemma \ref{Lem:ProperZone} can be extended as the following: there exists at least
one error vector $E_m$ such that $L(S+E_m)=L_k(S)$ and $supp(E_m)
\subseteq [a, a+2^s)$, for any $0\leq a \leq N-2^s$. In
Definition (\ref{Definitoin-k;Z-LC}), we just set the default zone
starting at 0.
\end{remark}
}


Because of the assumption $2^n - 2^s < L_k(S) < 2^n - 2^{s-1}$, we
derive $2^{s-1} < 2^n -L_k(S)  < 2^s$.  Let $Z = 2^{\lceil \log_2
(2^n-L_k(S)) \rceil}$. Then $L_{k;Z}(S)\leq L_k(S)$ by  Lemma
\ref{Lem:ProperZone}. Conversely, it is always true that $L_k(S)\leq
L_{k;Z}(S)$.  Therefore we obtain the following theorem.

\ChangeW{
\begin{theorem}\label{thm:EquivalenceofZoneLC}
Let $S$ be any $2^n$-periodic binary sequence. For any positive
integer $k$, there always exists $k^\prime\leq k$ such that
$k^\prime   \leq  Z=2^{\lceil \log_2 (2^n-L_k(S)) \rceil}$ and
$L_k(S)=L_{k^\prime}(S) = L_{k^\prime;Z}(S)$.
\end{theorem}
\begin{proof}
If $ k \leq  Z$,
 then $L_k(S)=L_{k;Z}(S)$. Otherwise if $k > Z$, then we can find an
 error vector $E_{k'}$ with the support in $[0, Z)$ such that $L_k(S) = L(S+E_{k^\prime}) $. The proof is complete because $L_{k^\prime}(S) = L_{k^\prime; Z}(S)$.
\end{proof}}

Theorem~\ref{thm:EquivalenceofZoneLC} shows that we can efficiently
verify the global stability of a binary sequence of period $2^n$
with large $k$-error linear complexity via a local stability.  If
$L_k(S)$ is big, then $Z$ can be very small.  If $k$ is large, then
we can reduce $k$ so that it is  bounded by the zone length as well.
\ChangeW{Of course, there is the degenerated case when $L_k(S)=0$, in this case,
we have to set $Z=2^n$. However, as we commented earlier, we focus on sequences with large $k$-error linear complexity and thus the zone length is significantly reduced.}

\section{Spectrum of $1$-error linear complexity with arbtrary zone length}
\label{Sec:Stabistic property}



In this section we assume $N=2^n$ and  $n\geq 4$. It is well known
that the linear complexity of a $2^n$-periodic sequence $S$ is $2^n$
if and only if it has odd Hamming weight. The 1-error linear
complexity of a $2^n$-periodic sequence can be any integer between
$0$ and $2^n-1$.
 However, if $S$ has an odd Hamming weight, then $L_1(S)$ can not be any integer of the form $2^n-2^s$ where $1\leq s \leq n$.
   For more details we refer the reader to \cite{SM,WMeidl,SM_PhDDis,SMWCC2009,SM_ISIT}.




If $S$ is a $2^n$-periodic sequence with even Hamming weight, then
$L_1(S) = L_0(S)$. In this case,  $L_{1, Z}(S) = L_0(S)$ for any
zone of length $Z$.  In order to study the distribution of
$L_{1,Z}$, we only need to consider $2^n$-periodic sequences with
odd Hamming weight.

\begin{theorem} \label{theorem_SeqOddHammingweight}
Let $S$ be a $2^n$-periodic sequence with odd Hamming weight and $L_1$ be its $1$-error linear complexity. Let $0<Z\leq 2^n$, 
and $a=\lfloor \log_2(Z) \rfloor$.


\begin{enumerate}
\item For $2^n-2^s < L_1 < 2^n-2^{s-1}$ with some integer $s$, we have the following

\begin{itemize}
\item[(i)] if $s\leq a$, then $L_{1;Z}(S)=L_1$, and the number of such
  sequences $S$ is $2^{L_1-1+s}$;\\
\item[(ii)]  if $s>a$, then we have
 \[ L_{1, Z}(S) \in \{ L_1, 2^n-2^{s-1}, \ldots, 2^n - 2^a\}.  \]

 The number of all the sequences \modify{$S$} that achieve these values
 equals
   \begin{eqnarray*}
 \left\{
  \begin{array}{ll}
   Z \cdot 2^{L_1-1},  &  if~ L_{1,Z} =L_{1,Z}(S) = L_1; \\
   Z\cdot 2^{L_1+s-t-2} ,  &  if~ L_{1,Z} = 2^n-2^t, where \; a+1 \leq t \leq s-1; \\
    2^{L_1-1+s}-Z\cdot 2^{L_1+s-a-2}, &  if~ L_{1,Z} = 2^n-2^a.
  \end{array}
             \right.
  \end{eqnarray*}
  \end{itemize}


\item    For $L_1=0$, we have

\[ L_{1, Z}(S) \in \{ 0, 2^n-2^{n-1}, \ldots, 2^n - 2^a\}.  \]

 The number of all the sequences \modify{$S$} that achieve these values equals
   \begin{eqnarray*}
 \left\{
  \begin{array}{ll}
   Z,  & if~ L_{1,Z} = L_1; \\
    Z \cdot 2^{n-t-1},  & if~ L_{1,Z} = 2^n-2^t, ~where~ a+1 \leq t \leq n-1; \\
    2^{n}-Z\cdot 2^{n-a-1}, & if ~L_{1,Z} = 2^n-2^a.
  \end{array}
             \right.
  \end{eqnarray*}

\end{enumerate}
\end{theorem}
\begin{proof}
For any $2^n$-periodic sequence $S$ with odd Hamming weight, we have
$2^n-2^s<L_1(S)<2^n-2^{s-1}$ for some positive integer $s$ or $L_1
=0$.  Let us first assume that $2^n-2^s<L_1(S)<2^n-2^{s-1}$ for some
positive integer $s$. From the
proof of  Lemma \ref{ErrorVectorsSpan}, 
 there exists exactly one $j$, $0\leq j< 2^s$, such that
$L_1(S) = L(S+E(j))$.

(i) if $s\leq a$, then $2^s\leq 2^a\leq Z$ by the definition of $a =
\lfloor \log_2(Z) \rfloor$.  Hence, in the zone of length $Z$, there
is at least one error vector reaching the $L_1(S)$, so $L_{1;Z}(S) =
L_1(S)$.

(ii) if $s > a$, then $Z  < 2^s$. Consider $0\leq j< 2^s$ such that
$L_1(S) = L(S+E(j))$. Let $m\cdot 2^{a+1} \leq j < (m+1)2^{a+1}$ for
some nonnegative integer $m$. Let $ \bar{j} = j -m 2^{a+1} $ be the
positive integer less than $2^{a+1}$ such that $\bar{j} \equiv j
\bmod{2^{a+1}}$.  If $ Z < \bar{j} < 2^{a+1}$, then we take
$j^\prime = \bar{j} - 2^a$, which satisfies $ 0<  j^\prime < 2^a <
Z$. Since $j^\prime-j$ is an odd multiple of $2^a$, we conclude that
$L(E_1(j) + E_1(j^\prime)) = 2^n- 2^a$ because the multiplicity of
the root $1$ for the binomial $x^{j} + x^{j^\prime}$ is exactly
$2^a$. Hence $L(S+E_1(j))= L_1(S) < 2^n -2^{s-1} \leq 2^n - 2^{a} =
L(E_1(j) + E_1(j^\prime))$. Therefore $L(S+E_1(j^\prime)) =
 \max\{L(S+E_1(j)), L(E_1(j) +
E_1(j^\prime))\} = 2^n -2^a$. And for any $\hat{j}$ satisfying $ Z <
\hat{j} < 2^{a+1}$, $\hat{j}\neq j^\prime$,
we have $L(E_1(j^\prime)+E_1(\hat{j}))>2^n-2^a$ because
$|\hat{j}-j^\prime| < 2^a$, accordingly we have
$L(S+E_1(\hat{j}))=\max\{L(S+E_1(j^\prime)), L(E_1(\hat{j})) +
E_1(j^\prime))\}>2^n-2^a$. Thus, $L_{1;Z} = 2^n-2^a$.

On the other hand, we assume  $0 \leq \bar{j} \ChangeMS{\leq} Z$.
For $m=0$, we must have $L_{1;Z}=L_1$ because $0\leq j < Z$. Let
$m\geq 1$.
We claim $L_{1;Z} = 2^n-2^{a+1+\nu_2(m)}$ where $\nu_2(m)$ is the
\emph{largest integer  such that $2^{\nu_2(m)} \mid m$}. Indeed, we
take  $j^\prime = j -m\cdot 2^{a+1}$. which satisfies $ 0< j^\prime
<  Z$. Since \Change{$j-j^\prime= m\cdot 2^{a+1}$}, we conclude that
$L(E_1(j) + E_1(j^\prime)) = 2^n- 2^{a+1+\nu_2(m)}$ because the
multiplicity of the root $1$ for the binomial $x^{j} + x^{j^\prime}$
is exactly $2^{a+1+\nu_2(m)}$. We note that $a+1+\nu_2(m) \leq s-1$
because \Change{$j- j^\prime= m2^{a+1} < 2^s$}. Hence $L(S+E(j))=
L_1(S) < 2^n -2^{s-1} \leq 2^n - 2^{a+1+\nu_2(m)} = L(E_1(j) +
E_1(j^\prime))$. Therefore $L(S+E_1(j^\prime)) = \max\{L(S+E_1(j)),
L(E_1(j) + E_1(j^\prime))\}  = 2^n -2^{a+1+\nu_2(m)}$. And for  any
$\hat{j}$ satisfying $ 0 \leq \hat{j} < Z$, $\hat{j}\neq j^\prime$,
we have $L(E_1(j^\prime)+E_1(\hat{j}))\geq 2^n-2^a$ because
$|\hat{j}-j^\prime| < 2^{a+1}$, accordingly we have
$L(S+E_1(\hat{j}))=\max\{L(S+E_1(j^\prime)), L(E_1(\hat{j})) +
E_1(j^\prime))\}\geq 2^n-2^a$.
 In this case, $L_{1;Z} = 2^n-2^{a+1+\nu_2(m)}$. Hence $L_{1, Z}(S) \in \{ L_1, 2^n-2^{s-1}, \ldots, 2^n - 2^a\}$.
We summarize these results in Table \ref{table:values} and Figure
\ref{figure:intervals}.

\begin{table}
\caption{$L_{1;Z}$ values at different intervals}
\label{table:values}
\[
\begin{array}{|c|c|c|}
 \hline
                    &  \ChangeMS{  0\leq j \pmod {2^{a+1}}   \leq Z}  & \ChangeMS{ Z <  j \pmod {2^{a+1}}  < 2^{a+1}} \\
\hline
0\leq j \leq 2^{a+1} &   L_1     &                2^n-2^a  \\
2^{a+1} \leq j < 2^{a+2} & 2^n - 2^{a+1} &        2^n - 2^a  \\
\vdots & \vdots & \vdots \\
m\cdot 2^{a+1}  \leq j < (m+1)2^{a+1} &  2^n - 2^{a+1 +\nu_2(m)} & 2^n-2^a \\
\vdots & \vdots & \vdots \\
2^s-2^{a+1} \leq j < 2^s & 2^n -2^{s-1} & 2^n-2^a \\
\hline
\end{array}
\]
\end{table}

Now we count the number of all these sequences having the $1$-error
linear complexity $L_1$ and odd Hamming weight such that $L_{1;Z}(S)
= 2^n - 2^t$ where $a+1 \leq t \leq s-1$. For each sequence $S$ with
$1$-error linear complexity $L_1$ and odd Hamming weight, we need to
count the number of error positions $j$'s such that
 $L_{1;Z} = L(S+E_1(j)) = 2^n -2^t$.
 We prove that the proportion of $j$'s over an interval of length $2^{t+1}$  such that $L_{1;Z} = 2^n -2^t$ is
 $Z/2^{t+1}$, where $t\geq a+1$.

First we show that every sub-interval $I$ of length $2^{t+1}$ in the
interval $[0, 2^s)$ contains at least one interval of length $Z$ for
possible $j$'s such that $L_{1;Z} = 2^n -2^t$. We will construct an
interval of length $Z$ for $j$'s within $I$. Because the length of
$I$ is $2^{t+1}$,  we can always choose an odd integer $m^\prime$
such that $[m^\prime 2^t,(m^\prime+1)2^t)
\subset I$.  
By the above proof, there exists an interval of length $Z$ within
the interval $ m^\prime 2^t = (m^\prime 2^{t-a-1}) 2^{a+1} \leq j <
((m^\prime 2^{t-a-1}) + 1)2^{a+1}$ such that $L_{1;Z} = 2^n -
2^{a+1+\nu_2(m^\prime2^{t-a-1})} = 2^n-2^{a+1+t-a-1} = 2^n -2^t$.

Then we show that $I$
can not contain more than one intervals of length $Z$ for possible
$j$'s such that $L_{1;Z} = 2^n -2^t$.  We prove it by contradiction.
Suppose there are $j, j^\prime$ such that $0<|j^\prime-j |< 2^{t+1}$
and $L(S+E_1(j)) = L(S+E_1(j^\prime)) = 2^n -2^t$. In this case,
$L(E_1(j) + E_1(j^\prime)$=$L(S+ E_1(j) + S+E_1(j^\prime))<2^n
-2^t$.  However,  \Change{the root  $1$ of $x^j + x^{j\prime}$ is }
at most $2^t$ times, implying $L(E_1(j) + E_1(j^\prime)) \geq 2^n
-2^t$,
a contradiction.

Therefore, the proportion of $j$'s in each interval of length
$2^{t+1}$ such that $L_{1;Z} = 2^n -2^t$ is $Z/2^{t+1}$, for each
sequence having the $1$-error linear complexity $L_1$ and odd
Hamming weight.   \Change{Since there are $2^{L_1 - 1 + s}$
sequences with odd Hamming weight such that $2^n - \ChangeGreen{2^s}
< L_1 < 2^n -2^{s-1}$ \ChangeGreen{(see \cite{SM_ISIT}[p. 2000,
Theorem 3])}, there are $Z2^{L_1-1+s}/2^{t+1} = Z 2^{L_1+s -t -2}$
sequences having the $1$-error linear complexity $L_1$ and odd
Hamming weight such that $L_{1;Z}(S) = 2^n - 2^t$.} Similarly, for
$L_{1;Z} = L_1$, there is only one internal of length $Z$ within the
interval $[0,2^{s})$ and thus the proportion is $Z/2^s$. From Table
\ref{table:values} and Figure \ref{figure:intervals},
the proportion of $j$'s giving $2^n - 2^{a}$ is
$1-\frac{Z}{2^{a+1}}$.

When $L_1=0$, the proof is similar and thus we omit the details.
     \end{proof}


\begin{figure}
\begin{center}
\includegraphics[bb=0 0 563 83, scale=0.6]{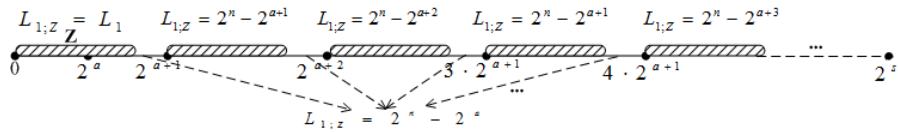}
\caption{$L_{1;Z}$ values for $j$ at different intervals}
\label{figure:intervals}
\end{center}
\end{figure}
%
The distribution of $k$-error linear complexity is provided in
\cite{SM_IEICE} when $k\leq 4$. In particular,  the number $N_0(c)$
of $2^n$-periodic sequences with the linear complexity $c$ is
$2^{c-1}$; see \cite{MNit}. In
Theorem~\ref{theorem_SeqOddHammingweight}, we have counted the
number of sequences with odd Hamming weight achieving $(1,Z)$-error
linear complexity values. In the following, we count the number
$\mathcal{N}_{1;Z}(c)$ of all sequences achieving $(1,Z)$-error
linear complexity value $c$, without emphasizing on their Hamming
weights.

\begin{corollary}
Let  $a=\lfloor \log_2(Z) \rfloor$. The value $\mathcal{N}_{1;Z}(c)$
is equal to  {\scriptsize
\begin{eqnarray*}
\!\!\!\!\!\!\!\!
\left\{
  \begin{array}{ll}
  2^{c-1}+2^{c-1+s}, & \!\!\!\!\!\!\!\!  \mbox{if}~s\leq a, 2^n-2^s< c< 2^n-2^{s-1};\\
  (1+Z)2^{c-1}, & \!\!\!\!\!\!\!\! \mbox{if}~s>a, 2^n-2^s <  c< 2^n-2^{s-1}; \\
   \displaystyle{\sum_{s=a+1}^n \sum_{L_1=2^n-2^s+1}^{2^n-2^{s-1}-1}(1-\frac{Z}{2^{a+1}})\cdot 2^{L_1-1+s}+(1-\frac{Z}{2^{a+1}})\cdot 2^n +2^{2^n-2^a-1}}, & \mbox{if}~c=2^n-2^a; \\
  \displaystyle{ \sum_{s=t+1}^n \sum_{L_1=2^n-2^s+1}^{2^n-2^{s-1}-1}   \frac{Z}{2^{t+1}}\cdot 2^{L_1-1+s}+ \frac{Z}{2^{t+1}}\cdot 2^n +2^{2^n-2^t-1}}, &  \mbox{if}~ c=2^n-2^t, a<t \leq s-1; \\
   1+Z, &  \mbox{if}~ c=0. \\
  \end{array}
             \right.
  \end{eqnarray*}
}
 \end{corollary}
\begin{proof}
We note that every $2^n$-periodic sequence $S$ has a linear
complexity $L(S) =c < 2^n$ if and only if it has an even Hamming
weight. In this case, $L_1(S) = L(S)$. Hence the result follows
immediately from Theorem \ref{theorem_SeqOddHammingweight}.
     \end{proof}

\change{The exact expectation $\mathcal{E}_{1;Z}$  can be derived from the above
counting functions, and may be used as a measure for determining the
randomness of a $2^n$-periodic binary sequence, with variations on
$Z$.  \ChangeW{ The exact formula is too complicated. Thus we omit all the details here. Instead, we  provide a concrete example for the expected values of  $L_{1;Z}$ for sequences with period $N=2^8$ in Table \ref{table:E1Z}.}

%
%
%
%
%
%
%
\begin{table}
\caption{ $\mathcal{E}_{N,1;Z}$ for $N=2^8$} \label{table:E1Z}
 \centering
 \begin{tabular}{|c|c|}
\hline
$Z$ & $\mathcal{E}_{N,1;Z}$ \\
\hline
1 &  254.0000\\
\hline
 2 & 253.5000\\
\hline
 3 & 253.2500\\
\hline
 4 & 253.0000\\
\hline
 5 & 252.9375\\
\hline
 6 & 252.8750\\
\hline
 7 & 252.8125\\
\hline
 8 & 252.7500\\
\hline
\end{tabular}
\end{table}

}

\section{Extension to sequences with other even periods } \label{sec:N=2^vr}
Now we consider stability of other periodic sequences with even
period $N=2^v r$ \modify{such that $r, v$ are  positive integer and
$r$ is odd}. For some types of $2^v r$-periodic binary sequences, we
can still find a proper zone of length $Z$ so that
$L_{N,k}=L_{N,k;Z}$. From the paper by Niederreiter \cite[Theorem 1,
P. 503]{Nied03}, there exists $2^v  r$-periodic binary sequence $S$
with $L_{N,k}(S)\geq N-2^v$ and $L(S)=N$, provided that
\begin{eqnarray*} \label{existence-condition}
\sum_{j=0}^k {N \choose j}< 2^{\min_{2\leq i \leq h} |C_i| },
\end{eqnarray*}
where $C_2,\ldots, C_h$ are the different cyclotomic cosets modulo
$r$. \modify{In the following,  we will reveal that
$L_{N,k}=L_{N,k;Z}$ with certain $Z \ll N$ for some of these `ideal'
cryptographic sequences.}

\begin{theorem} \label{2^vr-1-zonelength}
Let $N=2^v r$,
\revise{$v>0$}, $r$ be \ChangeMS{ an odd prime}, and $2$ be a
primitive root modulo $r$.
If $S$ is an $N$-periodic binary sequence such that
\begin{eqnarray} \label{Condition-Theorem2^vr}
L(S)=N-c>L_{N,k}(S)\geq N-\min(2^v,r-2),
\end{eqnarray}
for some nonnegative integer $c$, then there exists at least one
error vector $E_m$, $m\leq k$ such that
\[ L(S+E_m)=L_{N,k}(S) \mbox{ and } Supp(E_m) \subseteq [0, Z),\]
 where
$Z=2^{\lceil \log_2(N-L_{N,k}(S)) \rceil}$. In particular, if $L(S)= N > L_{N, 1}(S) \geq N- \min(2^v, r-2)$,  then there exists exactly one  error vector $E_1$
satisfying that  $Supp(E_1) \subseteq [0, Z)$ and
$L_{N,1}(S)=L_{N,1;Z}(S)$.

\end{theorem}
\begin{proof}
 Because $2$ is primitive root  modulo the prime number $r$, the cyclotomic polynomial $\phi_r(x)$ of the order $r$ is irreducible over $\mathbb{F}_2$. Hence $X^{2^vr}-1=(X^r-1)^{2^v} = ((X-1)\phi_r(X))^{2^v}$.
Let $\alpha$ be a primitive $r$-th root of unity. That is,
$\phi_r(\alpha)=0$.

 \ChangeMS{If
$2^v \leq r-2$}, then $\alpha$ can not be a root of $S^N(x)+E_t(x)$
for any error polynomial $E_t(x)$ of Hamming weight $t\leq k$
because of the assumption $L_{N,t}(S)\geq L_{N,k}(S)\geq N-2^v$.
Otherwise, $\phi_r(x) \mid s^N(x) + E_t(x)$ and thus the greatest
common divisor of $s^N(x) + E_t(x)$ and $X^N-1$ has degree greater
than $2^v$, a contradiction.

 \ChangeMS{If $2^v > r-2$}, then $L_{N, k}(S) \geq N- r+2$. Similarly,
$\alpha$ can not be a root of $S^N(x)+E_t(x)$  for any error
polynomial $E_t(x)$ of Hamming weight $t\leq k$.

Now we only need to consider the multiplicity of root $1$ when
computing $L_{N,k}(S)$.  { 
As in the proof of Lemma
\ref{Lem:ProperZone},  we can derive an error vector $E_m$ such that
$Supp(E_m) \subseteq [0, Z)$ and  $L_k(S+ E_m) = L_{N,k}(S)$.
Indeed, suppose there exists an error vector $E_m$  such that
$L(S+E_m) = L_{N, k}(S)$, where $Supp(E_m) = \{i_1, i_2, \ldots, i_m
\}$. Note that  $Z=2^{\lceil \log_2(N-L_{N,k}(S)) \rceil } \geq
N-L_{N, k}(S)$.   If $i_m \geq Z$, then we can define a new vector
$E_m^\prime = E_m + E_1(i_m) + E_1(i_m \bmod Z)$. Because $L(S+ E_m)
= L_{N, k} (S) \geq N-Z \geq L(E_1(i_m) + E_1(i_m \bmod Z))$, we
must have $L(S+E_m^\prime) \leq L_{N, k}(S)$ by counting the
multiplicity of $1$'s. Hence $L(S+E_m^\prime) = L_{N, k}(S)$.
Continuing this process, we can derive an error vector such that the
support is contained $[0, Z)$. }

{
In particular, if $L(S)= N > L_{N, 1}(S) \geq N- \min(2^v, r-2)$,  then there exists $s> 0$ such that
$2^{s-1}< {N-L_{N,1}(S)}  \leq   2^s$. From the previous discussion, $Z= 2^s$ and there exists at least one error vector $E_m$ such that $m \leq 1$ such that $L(S+E_m) = L_{N, 1}(S)$ and $supp(E_m) \subseteq [0, 2^s)$.  Because $L(S) > L_{N, 1}(S)$, we must have $m =1$.  Suppose there are $E_1,
E_1'$ such that $L(S+E_1) = L(S+E_1')=L_{N,1}(S)$ and $Supp(E_1),
Supp(E_1') \subseteq [0, 2^s)$.  In this case, the multiplicity of
root $1$ in $S^N(x) + E_1(x)$ and $S^N(x) + E_1'(x)$ is greater than
$2^{s-1}$ respectively, however the multiplicity of root 1 of the generating
polynomial corresponding to $E_1+E_1'$ is not more than $2^{s-1}$, a
contradiction. }
\end{proof}

Theorem~\ref{2^vr-1-zonelength} says that if $2^v \geq r-1$ then
$Z=2^{\lceil \log_2(N-L_{N,k}(S)) \rceil} \leq 2^{\lceil \log_2(r-2)
\rceil}  < 2r$. On the other hand,  if  $2^v < r-1$ then
Theorem~\ref{2^vr-1-zonelength} gives $Z \leq 2^v$ for any
$N$-periodic binary sequence $S$ such that  $L(S)=N-c \geq
L_{N,k}(S)\geq N-2^v$.

\begin{example}
Let $S_1$ be the following binary sequence with period $16*19$.
Namely,  $S_1=
11100011001111011010100011011100110000001100000001101110010100000\\
1101000110010010000100011001010010100110010111001001101011001000110100\\
0000100100001111001000101011000010010111110101001111110101101111000110\\
1001111101111110101011100100010001000011000011101111111111001011011011\\01001010011010010001001101000$.
The linear complexity is $304$ and $1$-error linear complexity is
$301$. The zone length is $4$ and we have
$L_{304,1}(S_1)=L_{304,1;4}(S_1)$.
\end{example}

{
\begin{example}
Let $S_2$ be a random generated binary
sequence with period $11*16$, $S_2=
11100111101100100110110111100100110101001010111111011011100010011\\
0001101001110010011101010100000100011111000001000011101011101101110011\\1010100011101000010000110 1101001101110010$. Then
$L(S_2)=L_1(S_2)=175$, $L_2(S_2)=169$, and the length of zone is $8$.
Indeed, we find an error vector of $L_2$  with two errors at
positions $6$ and $7$ within the zone $[0,8)$.
\end{example}
}

 We observe that the above result can be extended to
 any $N$-periodic binary sequence $S$ such that  $L(S)=N-c   \geq L_{N,k}(S)  =  N-r +1$ when $2^v < r-1$.   In this case, we can take $Z=r$.

\begin{proposition} \label{anotherResult}
Let $N=2^v r$, {$v>0$}, $r$ be \ChangeMS{an odd prime}, and $2$ be a
primitive root. If  $2^v < r-1$ and  $S$ is an $N$-periodic binary
sequence such that
 $L(S)=N-c   \geq L_{N,k}(S)  =  N-r +1 $ for some positve integer $c$, then
 there exists at least one error vector $E_m$, $m \leq k$
such that $L(S+E_m)=L_{N,k}(S)$ and $Supp(E_m) \subseteq [0, Z)$,
where  $Z=r$. In particular, $L_{N,k}(S)=L_{N,k;Z}(S)$.
\end{proposition}

\begin{proof}
Let $\alpha$ be the primitive $r$-th root of unity. Since the
multiplicity of 1 is at most $2^v < r-1$ and $L_{N, k}(S) = N-r+1$,
there exists an error vector $E_m$, $m\leq k$ reaching $L_{N,k}$,
such that the generating polynomial corresponding to $S+E_m$ will be
divisible by $\phi_r(x) = 1+X+\cdots+X^{r-1}$. Suppose $E_m$ has
entry `1'  at positions $i_1, i_2,\ldots, i_{m-1}, i_m$, where
$i_1<i_2<\ldots<i_m$. If $i_m>r$, since the generating polynomial of
$S+E_m+E_1(i_m)+E_1(i_m \bmod r)$ will be also divisible by
$1+X+\cdots+X^{r-1}$, \modify{  implying
$L\big(S+E_m+E_1(i_m)+E_1(i_m \bmod r)\big)\leq N-r+1$. Moreover,
$L(E_1(i_m)+E_1(i_m \bmod r) = N-r$ and $L_{N, k}(S) = L(S+E_m) \leq
N-r + 1$ imply that $L\big(S+E_m+E_1(i_m)+E_1(i_m \bmod r)\big)=
N-r+1$. }
 Thus we get an error vector $E_m+E_1(i_m)+E_1(i_m \bmod r)$ reaching
$L_{N,k}=N-r+1$. Consequently, we obtain an error vector reaching
$L_{N,k}(S)$ with support in $[0,r)$.
\end{proof}

\begin{remark}
From Theorem \ref{2^vr-1-zonelength} we obtain a small zone of
length $Z \leq 2^v$ or \ChangeRed{$Z< 2r$} such that
$L_{N,k}=L_{N,k;Z}$ for the above classes of sequences with large
linear complexity and $k$-error linear complexity. Our assumptions
on these classes of sequences are not very restricted. By Artin's
conjecture, there are
 approximately $37\%$ of all primes having $2$ as a
primitive root \cite{Moree}. By the following corollary
\ref{Cor:EventProb}, we show that under certain conditions almost
all random sequences have $k$-error linear complexity greater than
or equal to $N-2^v$ and about $50\%$ of these sequences have  linear
complexity equal to the period.  Therefore our result can be very
useful to determine the stability of many random binary sequences
with low computational cost.
\end{remark}

\begin{corollary} \label{Cor:EventProb}
Suppose $r$ is a prime, 2 is a primitive root modulo $r$ and
\modify{$2^v \leq p(r)$, where $p(r)$ is a polynomial of the
variable $r$}.
Let $N=2^v r$. If
\begin{equation} \label{k-requirement-equation}
\sum_{t=0}^k {N \choose t}< \frac{2^{r-1}}{2^v},
\end{equation}
\modify { for $0\leq c <\min(2^v, r-2)$, we have
\begin{eqnarray} \label{Pro-L(S)-LNK}
\Pr(L(S)\geq N-c, L_{N,k}(S)\geq N-\min(2^v, r-2) ) \rightarrow
1-2^{-1-c}, \mbox { as }  r \rightarrow \infty.
\end{eqnarray}
}
\end{corollary}
\begin{proof}
For any positive intger $k$, we denote by $\mathcal{M}_{N,k}(c)$ the
number of $N$-periodic sequences with the $k$-error linear
complexity not more than $c$. Obviously,
\begin{eqnarray*}\label{kErrorUpperbound}
\mathcal{M}_{N,k}(N-2^v-1) \leq \min \left( 2^N,
\mathcal{M}_{N,0}(N-2^v-1)\sum_{t=0}^k {N \choose t} \right).
\end{eqnarray*}

Then by Proposition 1 and Lemma 1 in \cite{MNit} (page 2818), we
have
$$\mathcal{M}_{N,0}(N-2^v-1)=2^N-(2^{r-1}-1)^{2^v}\cdot 2^{2^v}.$$

If $\mathcal{M}_{N,0}(N-2^v-1)\sum_{t=0}^k {N \choose t} \leq 2^N$,
 i. e.,
\begin{eqnarray} \label{sufficient condition}
\sum_{t=0}^k {N \choose t}\leq
\change{\frac{2^N}{\mathcal{M}_{N,0}(N-2^v-1)}=}  \frac{1}{1-
(1-\frac{1}{2^{r-1}})^{2^v}},
\end{eqnarray}
\ChangeMS{then we have $\mathcal{M}_{N,k}(N-2^v-1) \leq
\mathcal{M}_{N,0}(N-2^v-1)\sum_{t=0}^k {N \choose t}.$ }

Denote by $\rho$ the ratio  of the number of periodic sequences
satisfying $L_{N,k}(S) \geq N-2^v$ over the number of all periodic
sequences with period $N$. Hence
\begin{eqnarray} \label{roucondition}
\rho &=& \frac{2^N -  \mathcal{M}_{N,k}(N-2^v-1)}{2^N}  \nonumber \\
     &\geq & 1 - \frac{\mathcal{M}_{N,0}(N-2^v-1)}{2^N} \sum_{t=0}^k {N \choose t} \nonumber \\
     &= & 1-\left( 1- \big(1-\frac{1}{2^{r-1}}\big)^{2^v} \right)
\sum_{t=0}^k {N \choose t}.
\end{eqnarray}

Note that
\begin{eqnarray*}
\left(1-\frac{1}{2^{r-1}}\right)^{2^v}=\sum_{d=0\atop
s=2d}^{2^{v-1}-1}\Big( {2^v \choose s}\cdot \frac{1}{2^{(r-1)\cdot
s}} -  {2^v \choose s+1}\cdot \frac{1}{2^{(r-1)\cdot
(s+1)}}\Big)+\frac{1}{2^{(r-1)\cdot 2^v}}.
\end{eqnarray*}
For $d=0$ and $s=0$, we have  ${2^v \choose s}\cdot
\frac{1}{2^{(r-1)\cdot s}} -  {2^v \choose s+1}\cdot
\frac{1}{2^{(r-1)\cdot (s+1)}} = 1-\frac{2^v}{2^{r-1}}$. Then for
$d>0$,  we have $s \geq 2$ and
\begin{eqnarray*}
\frac{{2^v \choose s}\cdot \frac{1}{2^{(r-1)\cdot s}}}{{2^v \choose
s+1}\cdot \frac{1}{2^{(r-1)\cdot (s+1)}}} =
\frac{2^{r-1}(s+1)}{2^v-s}>\frac{2^r}{2^v}>1.
\end{eqnarray*}
The last inequality holds because \modify{ we have  $r>v$ by the
assumption}. Therefore, we obtain
\begin{eqnarray} \label{estimation}
\big(1-\frac{1}{2^{r-1}}\big)^{2^v}> 1-\frac{2^v}{2^{r-1}}.
\end{eqnarray}

If $\sum_{t=0}^k {N \choose t}< \frac{2^{r-1}}{2^v}$, then by
(\ref{estimation}) we have $\frac{2^{r-1}}{2^v}< \frac{1}{1-
(1-\frac{1}{2^{r-1}})^{2^v}}$ and the condition (\ref{sufficient
condition}) holds. Therefore, from (\ref{roucondition}) and
(\ref{estimation}) we derive
\begin{eqnarray*}
 \rho > 1-\frac{2^v}{2^{r-1}} \sum_{t=0}^k {N \choose t}.
\end{eqnarray*}
For small $k$,  we have $\sum_{t=0}^k {N \choose t} \leq c_0 N^k$
for some constant $c_0$, and thus\\
\modify{$\frac{2^v}{2^{r-1}}\sum_{t=0}^k {N \choose t} \leq 2^{1-r}
c_0 p(r)^{k+1} r^k$).}
Hence we must have
\begin{eqnarray*}\label{RouLimit}
\rho \rightarrow 1  \mbox{ as } r\rightarrow \infty.
\end{eqnarray*}

This implies that almost all sequences of period $2^v r$ satisfy
$L(S)\geq L_{N,k}(S) \geq N-2^v$ as long as $r \rightarrow \infty$.

\modify { Next we prove  that once $r\rightarrow \infty$, for $0\leq
c< 2^v$ we have
\begin{eqnarray}\label{Proportion-L(S)=N}
\frac{|S: L(S)=N-c|}{|S: L(S)\geq N-2^v|} \ChangeMS{ \rightarrow}
\frac{1}{2^{c+1}}.
\end{eqnarray}
 By the relationship between the linear complexity and G\"{u}nther weight of the GDFT of sequences
\cite{MNit}[p. 2818], and almost all sequences satisfy $L(S)\geq
N-2^v$, we only consider the first column of the GDFT matrix, and
the contribution to the G\"{u}nther weight of other columns are all
$2^v$. Additionally, the elements of the first column are over
$\F_2$, and the pattern of the first column is the transpose of
$\underbrace{0 \ldots 0}_{c} \; 1
*
*
*$, where `$*$' can be 0 or 1. Hence we have (\ref{Proportion-L(S)=N}).}



\ChangeMS{ If $2^v\leq r-2$},  then we obtain $\Pr(L(S)\geq N-c,
L_{N,k}(S) \geq N-2^v )= \Pr(L(S)\geq N-c, L(S) \geq N-2^v)
\ChangeMS{\rightarrow} \frac{1}{2} + \cdots + \frac{1}{2^{c+1}}  =
1-2^{-1-c}, \mbox { as } r \rightarrow \infty$.

\ChangeMS{If $2^v> r-2$},
then almost all sequences of period $2^v r$ satisfy $L(S)\geq
L_{N,k}(S) \geq N-2^v$ and we obtain $Pr(L(S)\geq N-r+2)\rightarrow
1-2^{1-r}$ similarly. Because
\begin{eqnarray*}
Pr(L_{N,k}(S)<N-r+2) \leq Pr(L(S)<N-r+2) \sum_{t=0}^k {N \choose t},
\end{eqnarray*}
\ChangeMS{for small $k$} we have
$Pr(L_{N,k}(S)<N-r+2)=O(2^{1-r}\cdot N^k)=O(2^{1-r}{p(r)}^k r^k)$,
then
$Pr(L_{N,k}(S)\geq N-r+2)\rightarrow 1$. Therefore, for $0\leq c
<r-2$ we obtain
\begin{eqnarray*}
\Pr(L(S)\geq N-c, L_{N,k}(S)\geq N-r+2 ) \rightarrow 1-2^{-1-c},
\end{eqnarray*}
analogously.
\end{proof}

\begin{remark}
According to the result in \cite{StevenRoman}[p. 25; Theorem 1.2.8],
for $0< \frac{k}{N}\leq 1/2$ we have
\begin{eqnarray} \label{Combinatorical-Inequality}
\sum_{t=0}^k {N \choose t} \leq 2^{N H(\frac{k}{N})},
\end{eqnarray}
where
$H(\frac{k}{N}):=-\frac{k}{N}
\log(\frac{k}{N})-(1-\frac{k}{N})\log(1-\frac{k}{N})$ is the entropy
function on the variable $\frac{k}{N}$, \change{and the base of the
$\log(\cdot)$ is 2.}

Consequently, if $k$ satisfies
\begin{eqnarray} \label{k-requirement-equation-explicit}
 N H(\frac{k}{N})<r-1-v,
\end{eqnarray}

\modify{then the condition (\ref{k-requirement-equation}) holds by
(\ref{Combinatorical-Inequality}) and
(\ref{k-requirement-equation-explicit})}. Hence we have a weaker but
explicit requirement of $k$ for (\ref{Pro-L(S)-LNK}) holds.
\change{Note that the entrophy function $H(x)$ is non-decreasing
when $0<x\leq 1/2$.}
\end{remark}

\ChangeW{ We provide the following example to demonstrate the usefulness of our results.
\begin{example}
Let $N=32*947$, i.e., $r=947$, $v=5$.  Let us consider $k=10$.
%
The traditional method of computing  $L_{N,10}(S)$ is estimated as
$\sum_{t=0}^{10} {32*947 \choose t} > {32*947 \choose 10} \approx
2^{127}$, which is infeasible. But from (\ref{k-requirement-equation-explicit}) we require
$H(\frac{k}{N})<\frac{941}{32*947}$ and thus $k\leq 96$. Hence   the condition (\ref{k-requirement-equation-explicit}) holds for $k=10$.
From Corollary \ref{Cor:EventProb} we know that almost all sequences satisfy $L_{N,10}\geq N-32$.  We can take the zone length  $Z=32$.

 If $c=1$ and $L(S)\geq N-c$, then  $3$ out of $4$ such random sequences satisfy the
condition (\ref{Condition-Theorem2^vr}) in Theorem
\ref{2^vr-1-zonelength}, and we have $L_{N,10}(S)=L_{N,10;32}(S)$.
If $L(S)\geq N-2$,  then $7$ out of $8$ such random sequences
satisfy (\ref{Condition-Theorem2^vr}), and we have
$L_{N,10}(S)=L_{N,10;32}(S)$. The percentage of these sequences
satisfying the condition (\ref{Condition-Theorem2^vr}) grows as $c$
increases. Finally,  if $L(S)\geq N-31$,  then almost any random
sequence satisfy
(\ref{Condition-Theorem2^vr}), and  $L_{N,10}(S)=L_{N,10;32}(S)$.
\end{example}
}

Now we move to other types of sequences with period $N=2^v r$, where
$r$ is a composite. More generally, we have
\begin{eqnarray} \label{cyclotomic-polynomial-decomposition}
\modify{1-X^N=(1-X^r)^{2^v}=\left(\prod_{d|r} \Phi_d(X)
\right)^{2^v}},
\end{eqnarray}
 where
$\Phi_d(X)$ is the $d$-th cyclotomic polynomial. \ChangeW{We do
not require that $\Phi_d(X)$ is irreducible over $\F_2$,  which is
required for existing fast algorithms of computing the ($k$-error)
linear complexity. Then, we can similarly obtain the following
result by analyzing the multiplicity of root 1 when the $k$-error
linear complexity is large.}

\begin{theorem}\label{2^vr-composite-extension}
Suppose $N=2^v p_1 ^{s_1} p_2 ^{s_2}\ldots p_n^{s_n}$,
 \change{$v>0$,} \ChangeMS{$p_i$ is odd prime},
and 2 is a primitive root modulo $p_i$ and $p_i^2$ for all $1\leq i
\leq n$. For any $N$-periodic binary sequence $S$ such that
\revise{$L(S)>L_{N,k}(S)\geq N-\min(2^v,p_1-2,\ldots, p_n-2)$,}
 there exists a zone of length $Z=2^{\lceil \log_2(N-L_{N,k}(S))
 \rceil}$ such that $L_{N,k}(S)=L_{N,k;Z}(S)$.

%
%
%
%
\end{theorem}
\begin{proof}
 First, we
suppose $N=2^v\cdot p^s$, where $2^v< p-1$ and $2$ is a primitive root modulo $p$ and $p^2$ respectively.  Obviously $2$ is a primitive root of $p^s$ for any
integer $s$ (see for example \cite{RC}).
From (\ref{cyclotomic-polynomial-decomposition}) we derive
\[
X^N-1=\left( (X-1)\Phi_p(X)\Phi_{p^2}(X)\cdots\Phi_{p^s}(X)
\right)^{2^v}.
\]

 Because the degree of each irreducible polynomial
\modify{$\Phi_{p^i}(X)$} is $\phi(p^i)=p^i-p^{i-1}\revise{\geq
p-1}$,
we only need to consider the multiplicity of root $1$ for estimating
$L_{N,k}$. { 
The rest of proof is similar to the proof
of Theorem \ref{2^vr-1-zonelength}.}

Secondly, suppose $N=2^v\cdot p \cdot q$, $2$ is a primitive root
modulo $p$ and $q$.
 Let $c$ be the least
integer such that $2^c \equiv 1 \bmod~({pq})$, then
\Change{$\Phi_{pq}(X)$} can be factorized into
$\frac{(p-1)(q-1)}{c}$ irreducible polynomials, each with degree
$c$. In addition, because $2$ is the primitive root modulo $p$, we
have $(p-1) \mid c$ and thus
\revise{$c\geq p-1$. Similarly,  $c\geq q-1$.}
From (\ref{cyclotomic-polynomial-decomposition}) we derive
$X^N-1=\left( (X-1)\Phi_p(X)\Phi_{q}(X)\Phi_{pq}(X) \right)^{2^v}$,
implying the degree of any irreducible factors except $X-1$ is
\Change{greater than or equal to} $\min(p-1,q-1)$. Hence we only need to
consider the multiplicity of root 1. {
The rest of proof follows.}

Finally, if $2$ is the primitive root of \ChangeRed{$p_1^{s_1}$,
$\ldots$, $p_n^{s_n}$} for any integers $i_1, \ldots, i_n$, then we
obtain \Change{ $X^N-1=\left( (X-1) \prod_{0\leq t_j \leq s_j \atop
j=1,2,\ldots,n} \Phi_{p_1^{t_1}\cdots p_j^{t_j}\cdots p_n^{t_n}} (X)
\right)^{2^v}$}. Similarly, the degree of each irreducible factor of
\Change{$\Phi_{p_1^{t_1}\cdots p_j^{t_j}\cdots p_n^{t_n}}(X)$} is
\revise{no less than $\min(p_1-1, \ldots, p_n-1)$.  Hence we only need to
consider the multiplicity of the root 1 analogously.}
\end{proof}

\begin{example} From computer experiments, there are many examples satisfying the above theorem for sequences with period $N$ such as $N=32*37, \change{4*11^2},  8*11^2, 8*11*13$;
\change{$16*11^2, 16*11*13$; \modify{11*16*19, 13*16*19.}} The
global stability can be effectively determined by local stability
within a much smaller zone.  For example, for the random generated
$N=4*11^2$-periodic sequence\\
$S=010110001110010010101001101011100110010010001110011100010110001000\\
0010100100100101001000011110000101011100011110000000111010011001101110\\
0110101110100010110111000011101010100101101100001000101001111010010110\\
1110000000011010111010000000100100001011011100000100111001110101110111\\
0010101000100111010010100011110011110100110110001000110000010001110100\\
0110101010110001000100000000000001011011001111000011010000100111110000\\
00101110000000100000111011001100110100111101000000100100100110000100$,
we have $L(S)=L_1(S)=482$, $L_2(S)=480$, and the zone $[0, 4)$.
Indeed, we can find an error vector \change{ $E_2$ of $L_2$ such
that error positions are  $1$ and $3$ } within the zone $[0,4)$.
\end{example}

{ 
\begin{remark}
For a binary sequence $S$ with the period $N=2^v m$, $m\geq 1$,
$v>0$. If $L_{N,k}(S)\geq \frac{N}{2}$, then the derived length of
zone in Theorems \ref{thm:EquivalenceofZoneLC},
\ref{2^vr-1-zonelength}, and \ref{2^vr-composite-extension}
\begin{displaymath}
Z=2^{\lceil \log_2(N-L_{N,k}(S)) \rceil}\leq 2^{\lceil \log_2(N/2)
\rceil}= 2^{\lceil \log_2(N)-1 \rceil}<N,
\end{displaymath}
so $Z$ becomes effective: the larger $L_{N,k}(S)$, the smaller $Z$.
\end{remark}
}

\ChangeRed{
\section*{Acknowledgment}
Ming Su expresses his sincere thanks for the hospitality during his
visit to School of Mathematics and Statistics, Carleton University,
Canada. }

\end{document}